\documentclass[11pt,twocolumn]{article}
\usepackage[centertags]{amsmath}
\usepackage{amsfonts}
\usepackage{amssymb}
\usepackage{amsthm}
\usepackage{amsgen}
\usepackage{amscd}
\usepackage{graphicx}
\usepackage{algorithmic}
\usepackage{algorithm}
\usepackage{multirow}
\usepackage{extramarks}
\usepackage{fancyhdr}
\usepackage{subfigure}

\input{xy}
\xyoption{all}
\theoremstyle{plain}
\newtheorem{thm}{Theorem}[section]
\newtheorem{cor}[thm]{Corollary}
\newtheorem{lem}[thm]{Lemma}

\newtheorem{defn}[thm]{Definition}
\theoremstyle{remark}
\newtheorem{rem}{Remark}[section]
\numberwithin{equation}{section}
\numberwithin{algorithm}{section}

\begin{document}


\title{\textbf{STRUCTURE AND DYNAMICS OF POLYNOMIAL DYNAMICAL SYSTEMS
}\footnote{
This work was supported by the National Science Foundation under Grant Nr. CMMI-0908201.}}

\author{\hspace{3 in}\textbf{Reinhard Laubenbacher}\\\hspace{3.4 in}Virginia Bioinformatics Institute\\\hspace{2.5in} \textbf{David Murrugarra}\\ \hspace{3.4 in}Virginia Bioinformatics Institute\\\hspace{2.3 in}
\textbf{Alan Veliz-Cuba}\\\hspace{3.3 in}University of Nebraska Lincoln}
\date{}
\maketitle
\thispagestyle{fancy} 

\begin{center}
Abstract
\end{center}
Discrete models have a long tradition in engineering, including finite state machines,
Boolean networks, Petri nets, and agent-based models. Of particular importance is the
question of how the model structure constrains its dynamics. 
This paper discusses an algebraic framework to study such questions.  The systems discussed here
are given by mappings on an affine space over a finite field, whose coordinate functions are polynomials. 
They form a general class of models which can represent many discrete model types. 
Assigning to such a system its dependency graph, that is, the directed graph that indicates the variable
dependencies, provides a mapping from systems to graphs. A basic property of this mapping is derived
and used to prove that dynamical systems with an acyclic dependency graph can only have a unique
fixed point in their phase space and no periodic orbits. This result is then applied to a published model
of in vitro virus competition.

\section{Introduction}
In several areas of engineering, discrete mathematical models, such
as finite state machines, Boolean networks, Petri nets, or agent-based
models, play an
important role in modeling processes that can be viewed as evolving
in discrete time, in which state variables have only finitely many
possible states. Decision processes, electrical switching networks, 
or intra-cellular molecular networks represent examples.  There is a long
tradition in the engineering literature of work related to an understanding
of how the structure of such models constrains their dynamics, 
see, e.g., \cite{Elspas, Butler}. This problem is important both for the design
of networks as well as for their analysis. 
In both theoretical and applied studies of dynamical systems, the problem
of predicting dynamic features of the system from structural properties
is very important, see, e.g., \cite{Golubitsky, Soule}. An instantiation
of this problem is the question in systems biology whether the 
connectivity structure of molecular networks, such as metabolic or
gene regulatory networks, has special features that correlate with
the type of dynamics these networks exhibit. (This type of application
motivated the present study.) Some progress has
been made on this question. For instance, in \cite{Alon} it was
shown that certain small network motifs appear much more often
in such networks than would be expected in a random graph. 
The work in this paper was motivated by the desire to provide a
theoretical framework within which the relationship between
structure and dynamics of certain families of dynamical systems
can be studied. We briefly describe this framework and show one
example of a result and its implications. 

The type of dynamical system studied here is discrete in time as well
as in variable states. That is, we consider a collection $x_1, \ldots , x_n$ of
variables, each of which can take on values in a finite set $X$.
We consider time-discrete dynamical systems
$$
f: X^n\longrightarrow X^n.
$$
A well-known instantiation of such systems are Boolean networks,
which have many applications in molecular biology as well as 
engineering. In this case, $X = \{0, 1\}$.
As described, such systems are set functions without
additional mathematical structure. It is therefore advantageous to
impose additional mathematical structure on $X$, namely that of a 
finite field. This is of course utilized in the case of Boolean networks.
Evaluation of Boolean functions is equivalent to carrying out 
arithmetic in the Boolean field with two elements. The translation
between Boolean functions and polynomials 
is straight-forward, with the Boolean operator AND corresponding
to multiplication, addition corresponding to XOR, and negation
corresponding to addition of $1$.  Once we have a finite field structure
on $X$, which we shall now denote by $\mathbb K$, it is well-known
that any function $\mathbb K^n \longrightarrow \mathbb K$ can
be represented as a polynomial function \cite[p. 369]{Lidl}. Thus, we can focus
on systems
$$
f = (f_1, \ldots , f_n): \mathbb K^n\longrightarrow \mathbb K^n,
$$
for which each $f_i$ is a polynomial function in the variables $x_1, \ldots , x_n$.
We will call such a system a polynomial dynamical system (PDS) over $\mathbb K$.
A powerful consequence is that one can use algorithms and theoretical
results from computer algebra and computational algebraic geometry, such
as the theory of Gr\"obner bases. 

 It has been shown \cite{Veliz-Cuba, Hinkelmann} that discrete models in several different
 frameworks can be translated into the PDS framework, namely $k$-bounded Petri nets,
 so-called logical models used in molecular biology, and many agent-based models,
 which are becoming very prominent in biology. Thus, the framework of PDS provides
 provides a common theoretical formulation. Using algorithms from polynomial algebra,
 the software package ADAM (Analysis of Dynamic Algebraic Models) \cite{Hinkelmann2} 
 is very efficient in analyzing various features of PDS, including their dynamics. 
 Thus, any theoretical results about PDS apply to all these model types. 
 
 In this paper, we are concerned with the family $\mathcal P$ of all PDS of a given dimension $n$ over
 a finite field $\mathbb K$. The structural information for a PDS includes a directed
 graph that indicates the dependency relationships between variables. In the context
 of a molecular network model, this graph would represent the wiring diagram of the
 network. Graphs can be represented by their adjacency matrices. The set of $n\times n$
 matrices carries an algebraic structure using max and min operations, which has
 been studied previously over the real field, see, e.g., \cite{Gavalec}. 
 Here we prove a result about the relationship between properties
 of dynamical systems in $\mathcal P$ under composition and properties of the
 corresponding adjacency matrices in the max-min algebra.  As a consequence of
 this result, we derive a dynamic property of PDS with acyclic dependency graphs. 
 Finally, we give an application of this consequence to a published model of {\it in vitro}
 virus competition. 

\section{The algebra of dynamical systems}

Let $\mathbb{K}$ be any finite field. We consider dynamical systems
$$
f=(f_1,\ldots,f_n):\mathbb{K}^n\rightarrow
\mathbb{K}^n,
$$
where
$f_i:\mathbb{K}^n\rightarrow \mathbb{K}$ for $i=1,\ldots,n$.
As observed above, any such $f$ is a PDS. 
Let $\mathcal P$ denote the family of all such systems. It is well-known
that $\mathcal P$ has the structure of an associative algebra under
coordinate-wise addition and composition of functions. 

To $f$ we can associate a directed graph with the $n$ nodes 
$x_1, \ldots , x_n$. There is a directed edge from 
$x_j$ to $x_i$ if $x_j$ appears in $f_i$. Let $[f]$ be the
$n\times n$ adjacency matrix of this graph. That is, 
$[f]=(a_{ij})$ is defined as follows:
\begin{displaymath}
a_{ij}=\left\{
\begin{array}{ll}
 1 &f_i\text{ depends on } x_j  \\
  0 &\text{otherwise}   
\end{array}
\right.
\end{displaymath}

Equivalently,
$a_{ij}=1$ if and only if there are $p\neq q\in\mathbb{K}$ such that
\begin{displaymath}
f_i(x^{(j,p)}) \neq f_i(x^{(j,q)})
\end{displaymath}
where $x^{(j,p)} = (x_1,\dots,x_{j-1},p,x_{j+1},\dots,x_n)\in\mathbb{K}^n$.

The adjacency matrix has binary entries, and we now define two operations on 
such matrices. Let $\mathbb{B}$ denote the Boolean field with two elements,
with the natural order $0 < 1$. 

\begin{defn}
Given $a,b\in\mathbb{B}$, we let:
\begin{description}
\item[1)] $a\oplus b=\max\{a,b\}$ and
\item[2)] $a\star b=\min\{a,b\}$
\end{description}
\end{defn}

\begin{defn}
Given $A=(a_{ij})$ and $B=(b_{ij})$ matrices with entries in
$\mathbb{B}$, we define
the following operation:\\
$A\star B=C=(c_{ij})$, where $c_{ij}=\bigoplus^n_{k=1}(a_{ik}\star
b_{kj})$
\end{defn}

\begin{rem}
$c_{ij}=1$ if and only if there is $k$ such that $a_{ik}\star b_{kj}=1$
\end{rem}

\begin{defn}
We define $A\preceq B$ if and only if $a_{ij}\leq b_{ij}$ for all
i,j.
\end{defn}

Forming the adjacency matrix of the dependency graph of a PDS
then gives a mapping
$$
\mathcal{P}\longrightarrow \mathcal{M}. 
$$And we also have a mapping that
associates to an element in $\mathcal{P}$ its phase space, a directed
graph on the vertex set $|\mathbb{K}^n|$, which encodes the
dynamics of the system. We will denote by $\mathcal{S}$
the set of all directed graphs on $|\mathbb{K}^n|$ with the property that
each vertex has a unique outgoing edge (the requirement for being the
phase space of a deterministic PDS). Hence, we have mappings
$$
\mathcal{S}\longleftarrow \mathcal{P}\longrightarrow \mathcal{M}.
$$
The result in the next section relates information in $\mathcal{M}$ to
information in $\mathcal{S}$.
\section{A property of $\mathcal{P}\longrightarrow\mathcal{M}$}

The main result of this paper is the following theorem. It describes a basic property
of the mapping that extracts the adjacency matrix from the system. 

\begin{thm}\label{thm1}
We have that
$[f\circ g]\preceq[f]\star[g]$ in $\mathcal{M}$ for all $f,\ g\in \mathcal{P}$.

\end{thm}

\begin{proof}

Let $A$, $B$, and $C$ be the adjacency matrices for $f$, $g$, and $f\circ g$, respectively.
It suffices to prove that, if $c_{ij}=1$, then $a_{ik}b_{kj}=1$ for some $k$.
Let us assume that $c_{ij}=1$, then there are $p\neq q\in\mathbb{K}$ such that
\begin{displaymath}
(f\circ g)_i(x^{(j,p)}) \neq (f\circ g)_i(x^{(j,q)}),
\end{displaymath}
where $x^{(j,p)} = (x_1,\ldots,x_{j-1},p,x_{j+1},\ldots,x_n)\in\mathbb{K}^n$.
Now, since $(f\circ g)_i = f_i\circ g$ we have that
\begin{displaymath}
f_i(y_1,\dots,y_n) \neq f_i(\hat{y}_1,\dots,\hat{y}_n),
\end{displaymath}
where $y_s = g_s(x^{(j,p)})$ and $\hat{y}_s = g_s(x^{(j,q)})$ for $s = 1,\dots,n$. 
Then there is an index $k$ such that $f_i$ depends on $x_k$ and
\begin{displaymath}
y_k\neq\hat{y}_k,
\end{displaymath}
or
\begin{displaymath}
g_k(x^{(j,p)})\neq g_k(x^{(j,q)}).
\end{displaymath}
Therefore, $a_{ik}=1$ and $b_{kj}=1$ and, consequently, $a_{ik}b_{kj}=1$.
\end{proof}

\begin{cor}\label{ImporCor}
$[f^r]\preceq[f]^r$ for all $f\in \mathcal{P}$, and for any $r\geq1$.
\end{cor}

\begin{proof}
Replace $g$ by $f$ in Theorem~\ref{thm1}.
\end{proof}

\section{PDS with acyclic dependency graph}
As an easy application of our main result we now consider PDS in $\mathcal{P}$
that have an acyclic dependency graph, that is, no feedback loops in their structure.
For any $f:\mathbb{K}^n\rightarrow \mathbb{K}^n$ with acyclic dependency
graph  we can see easily that its adjacency matrix $[f]$ is a strictly triangular matrix, i.e., with
zeros in the diagonal (otherwise, there would be loops in the graph). Therefore, its
characteristic polynomial is equal to $\lambda^r$, where $r$ is the order of
the matrix. So $[f]$ is nilpotent and $[f]^r=0$.

\begin{cor}\label{main}

Any discrete dynamical system $f:\mathbb{K}^n\rightarrow \mathbb{K}^n$ with acyclic dependency graph has a unique fixed point.

\end{cor}

\begin{proof}

By Corollary \ref{ImporCor}, $[f^r]\preceq[f]^r=0$, so $f^r$ is constant:
$f^r\equiv x_0$. Therefore, $x_0$ is the unique fixed point of $f$.

\end{proof}

Note that the nilpotency index of $[f]$, the smallest integer $r$ for which $[f]^r=0$ and $[f]^{r-1}\neq0$, 
gives us an upper bound for the number of steps to reach the unique fixed point from any other state.

\section{Application}

In this section we will present an application of Corollary~\ref{main} to a published model
of an \textit{in vitro } competition between two strains of a murine coronavirus studied in Jarrah et. al.~\cite{Jarrah}. The model is presented as a hexagonal grid of cells, with color coding of cells to indicate their infection status. Normal cells are represented as white. Infected cells are represented as red or green, or yellow, in the case of dual infection.
The infection spreads from the center outwards. At each time step one ring of new cells is infected. 
The outcome of a cell in the new ring depends on the infection status of its two neighbors in the 
previous infected ring. The local update function for each cell is constructed using the following rules:
  
  \begin{itemize}
  \item If a cell has only one infected neighbor, it will get the same type of infection.  
  \item If a cell has two infected neighbors, then we use the following table to determine the type of infection of that cell.
\end{itemize}

\begin{center}
\begin{tabular}{| c | c | c | c | c |}
\hline
\multicolumn{5}{| c | }{Rules for the update function} \\
\hline
& Green & White & Red & Yellow \\ \hline
Green & Green & Green & Yellow & Green\\ \hline
White & Green & White & Red & Yellow \\ \hline
Red & Yellow & Red & Red & Red \\ \hline
Yellow & Green & Yellow &Red & Yellow \\ \hline
\end{tabular}
\end{center}

The dynamics of this system is represented in Figure~\ref{Evolution_without_controllers}.
In this example we will use 169 cells. In order to simulate experimental conditions described in \cite{Jarrah}
we initialize the infection using the 37 center cells. Each cell can have only one of four colors at a time. We use the field
with four elements, $\mathbb{F}_4 = \{0, 1, 2, 3\}$, to represent the set of different colors. 
The color assignment is as follows. 

\begin{center}
\begin{tabular}{| c | c | }
\hline
\multicolumn{2}{| c | }{Color assignment} \\
\hline
Color&Field element\\ \hline
Green&0\\ \hline
Red & 1 \\ \hline
White& 2\\ \hline
Yellow&3\\ \hline
\end{tabular}
\end{center}

We represent the 169 cells by the variables $x_1,\dots,x_{169}$, with $x_1,\dots,x_{37}$ representing the center cells. The variables $x_{128},\dots,x_{169}$ represents the cells in the outermost ring. 

In \cite{Jarrah} the system was studied from the point of view of the experimental system. There, a collection
of cells was infected in the center of the dish, and the infection was then observed to spread to the rest of
the Petri dish in a pattern that showed distinct segmentation, matching the experimental results. In order
to mimic the biological system, it was assumed that the cells initially infected did not change their infection
state subsequently, so that the model is constrained and heterogeneous with respect to the rules assigned
to all the nodes. 
This, in effect, changes the dynamical system since those cells are now assigned
constant functions rather than the rules described above. The outcome is a steady state that shows 
a characteristic segmentation which coincides with experimental observations. 
See Figure \ref{Evolution_without_controllers} for an example.
\begin{figure}
\begin{center}
\includegraphics[width=3in]{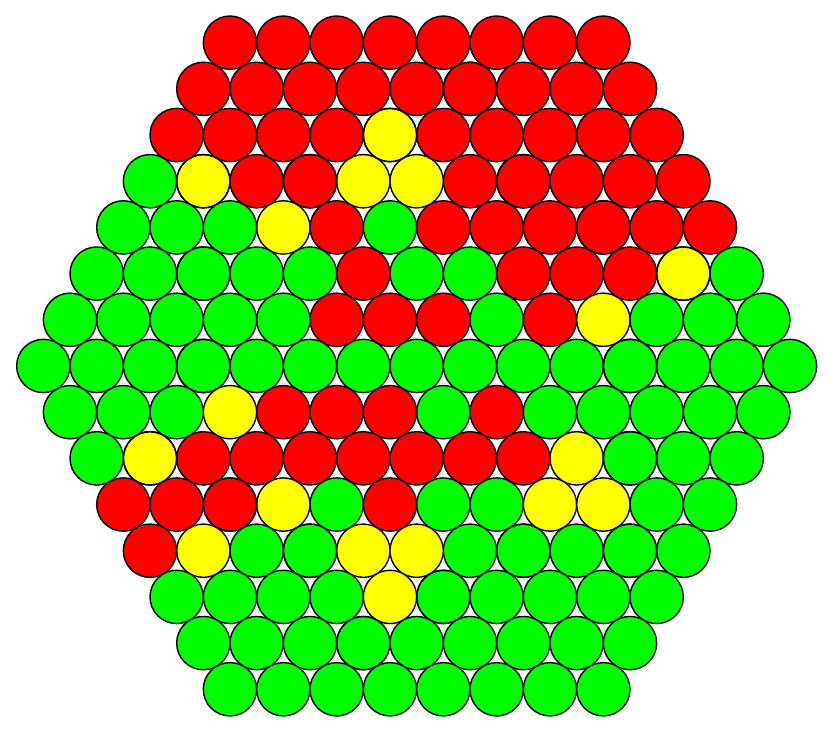}
\caption{Dynamics of the constrained model}
\label{Evolution_without_controllers}
\end{center}
\end{figure}

If, however, we allow all cells to evolve according to the update rule above, then the
outcome is significantly different. In that case, the center cell plays a very special role. It is represented by
the only node in the network that does not receive any inputs from other cells and is therefore constant. (Note that, in order
to have an acyclic dependency graph, a PDS must have at least one node without incoming edges.) The
state of that node then plays the role of an external parameter, and a particular choice of state/color for
this node is part of the system description. 
In this case Corollary \ref{main} applies, since the dependency graph of the network is acyclic. 
See Figure  \ref{AdjacencyGraph}.
\begin{figure}
\begin{center}
\includegraphics[width=3in]{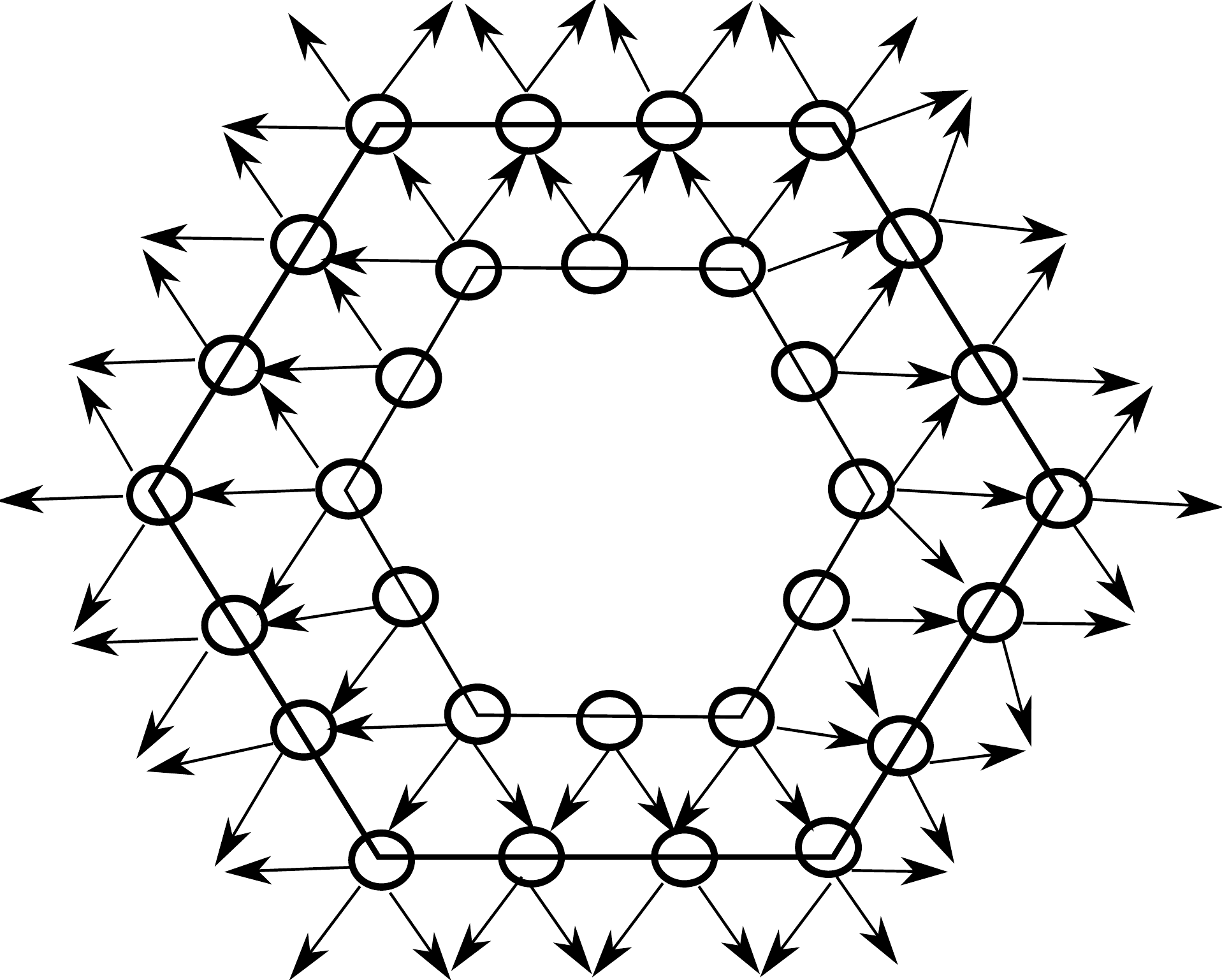}
\caption{Adjacency graph for virus competition }
\label{AdjacencyGraph}
\end{center}
\end{figure}
We conclude that in this case any initialization of the system, that
is, any assignment of colors to each of the nodes in the network, results in a unique steady state, namely the
state in which all nodes have the same color as the center node. In effect, what happens is that after initializing
the 37 nodes in the center they propagate a distinct, segmentation-like pattern. At the same time, the color
in the center cell propagates outward and overpowers the segmentation pattern, resulting in a system state
with a homogenous color distribution. See Figure  \ref{Evolution}. 

If we now assume that a larger number of cells in the center is infected, that is, is assigned a color that does not change
subsequently, as in the initial experiments, then the color distribution along the edge of the infected region propagates
and produces a steady state that show the segmentation patterns observed in \cite{Jarrah}.

The polynomial dynamical system for the model in which all cells are allowed to change is represented as  
\begin{displaymath}
f=(f_1,\dots,f_{169}):\mathbb{F}_4^{169}\rightarrow \mathbb{F}_4^{169}
\end{displaymath}
The coordinate functions $f_i$ are polynomials in $\mathbb{F}_4[x_j,x_k]$, where $x_j$ and $x_k$ are the two neighbors of $x_i$ in the previous infected ring. 
Let $x_i$ represent one of the 169 cells, then $x_i$ is updated according $f_i(x_j,x_k)$, i.e. $x_i=f_i(x_j,x_k)$, where $x_j$ and $x_k$ are the two neighbors of $x_i$ in the previous infected ring. The following table specifies part of the truth table for $f_i(x_j,x_k)$,
\begin{center}
\begin{tabular}{| c | c | c | c | c | c | }
\hline
\multicolumn{6}{| c | }{Truth Table} \\
\hline
Color&Color&$x_j$ &$x_k$ &$f(x_j,x_k)$&Color\\ \hline
Green& Green&0 &0 &0&Green\\ \hline
Green &Red  & 0&1 &3&Yellow\\ \hline
Green& White & 0&2 &0&Green\\ \hline
Green&Yellow & 0&3 &0&Green\\ \hline
Red&Green&1&0&3&Yellow\\ \hline
Red&Red&1&1&1&Red\\ \hline
Red&White&1&2&1&Red\\ \hline
Red&Yellow&1&3&1&Red\\ \hline

\end{tabular}
\end{center}

Written as a polynomial, $f_i$ has the following form
  \begin{displaymath}
  \begin{array}{lcl}
f_i(x_j,x_k)&=& x_j + 3 x_j^2 + x_j^3 + 3 x_j^4 + x_k + 3 x_j x_k +  \\ \\
     & & x_j^2 x_k + x_j^3 x_k + 2 x_j^4 x_k + 3 x_k^2 + x_j x_k^2 +  \\ \\
      && 4 x_j^2 x_k^2 + 4 x_j^4 x_k^2 + x_k^3 + 4 x_j^2 x_k^4 +\\ \\
      &&  3 x_j^3 x_k^4 + 3 x_j^4 x_k^4
\end{array}
\end{displaymath}
This polynomial form can then be used to compute the steady state configuration by
solving a system of polynomial equations. 

\begin{figure*}[h!tp]
  \begin{center}
    \subfigure
      {\includegraphics[width=.3\textwidth]{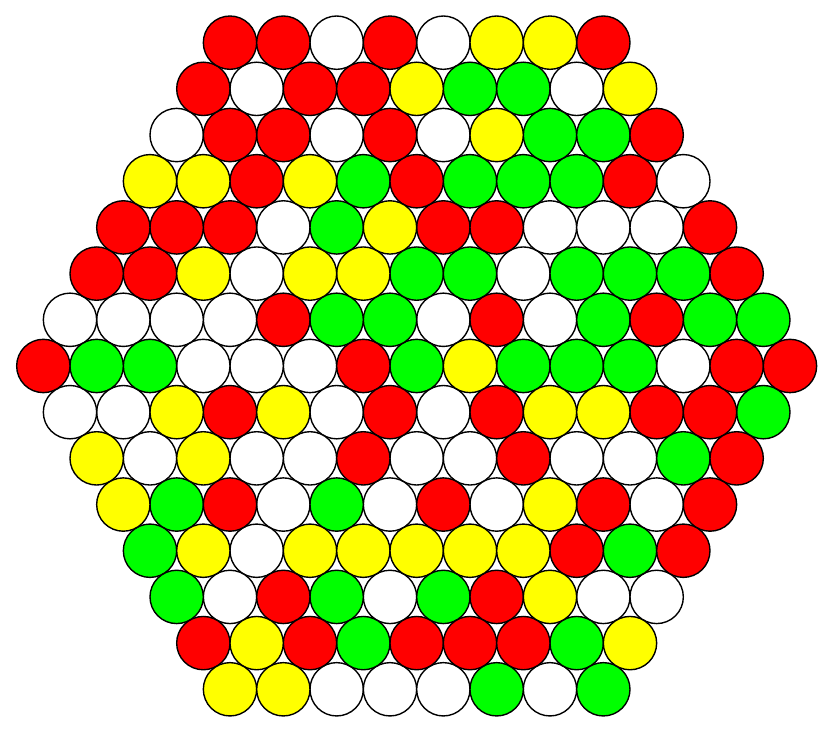}\includegraphics[width=.3\textwidth]{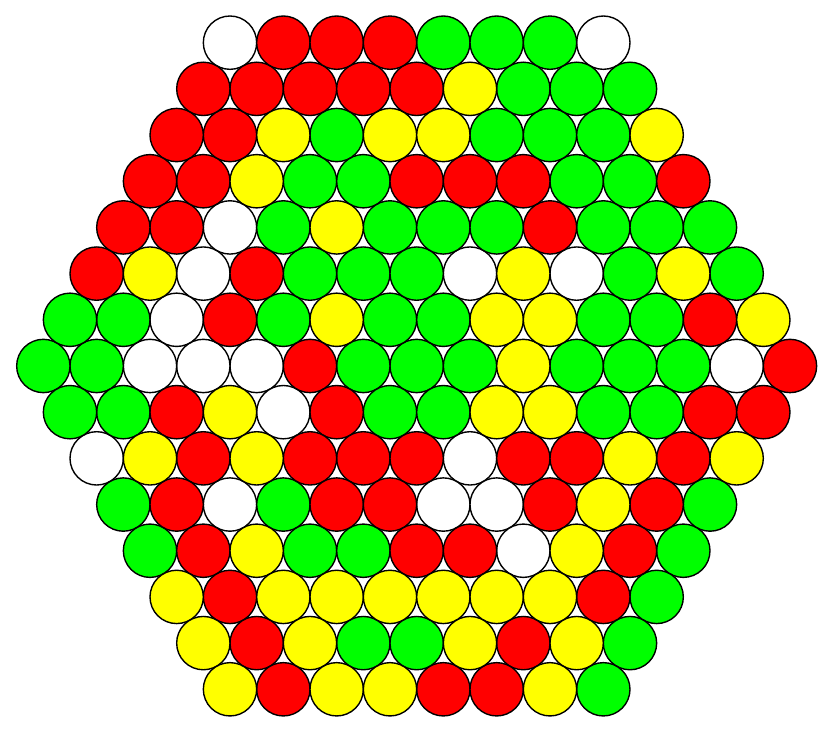}
      }
    \subfigure
      {\includegraphics[width=.3\textwidth]{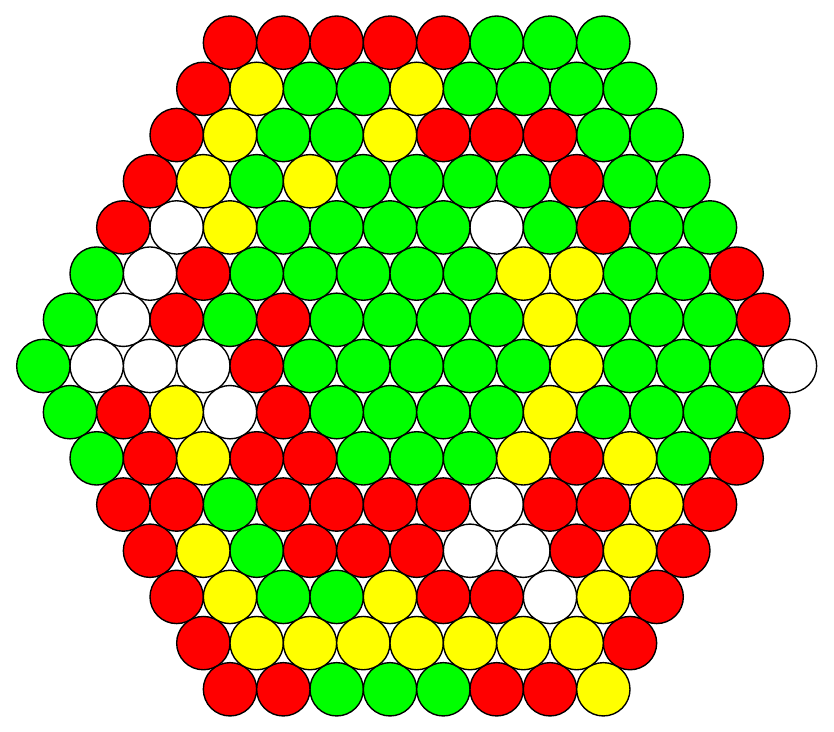}
        \includegraphics[width=.3\textwidth]{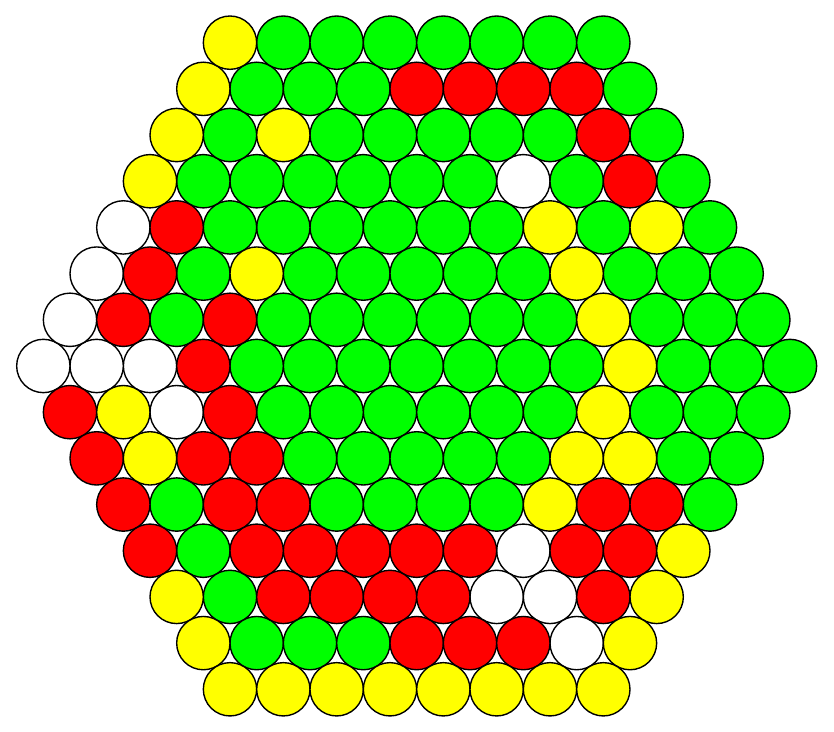}
      }
      \subfigure
      {\includegraphics[width=.3\textwidth]{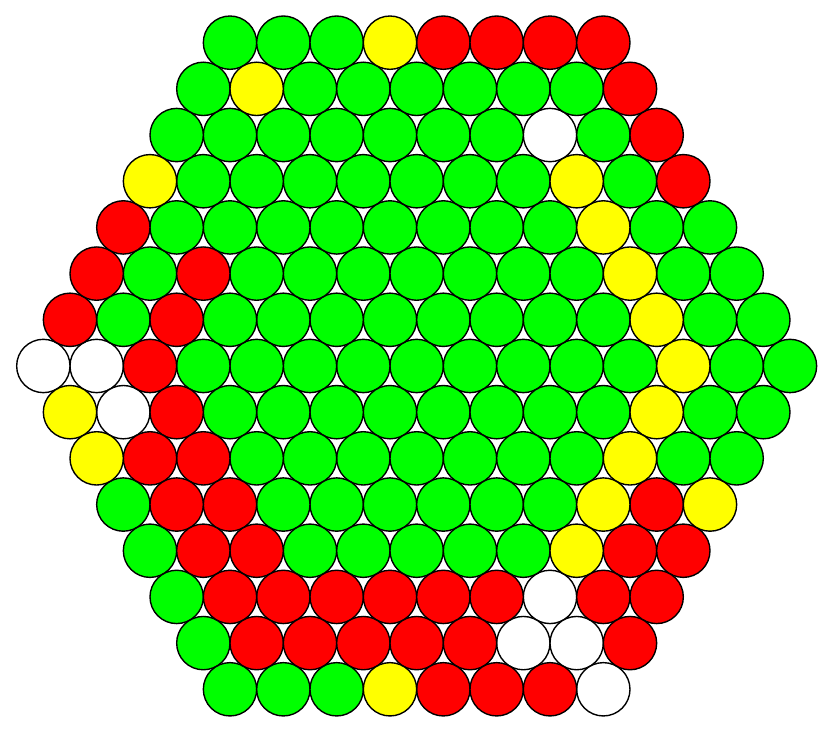}
        \includegraphics[width=.3\textwidth]{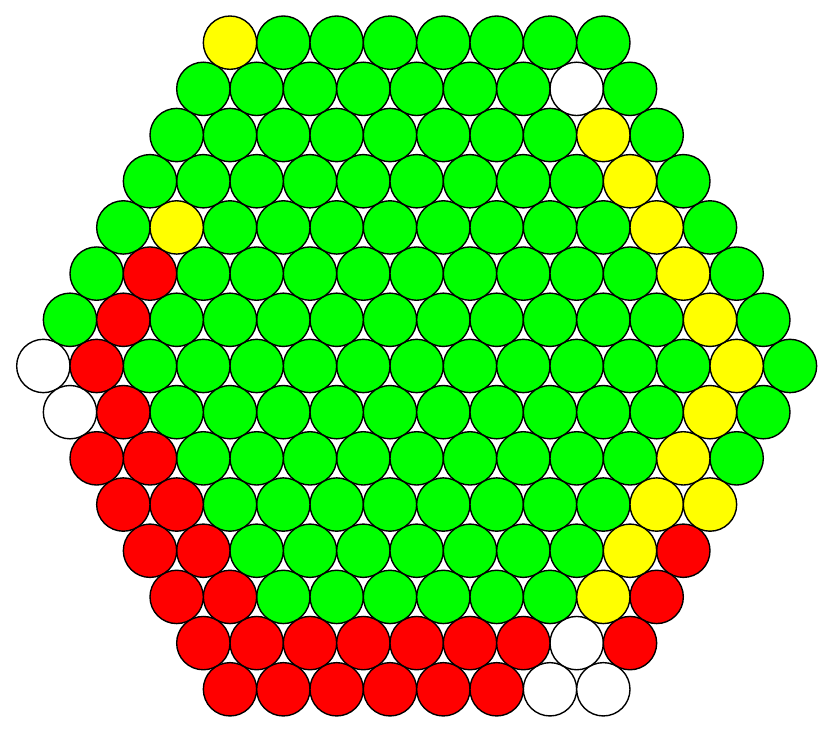}
      }
      \subfigure
      {\includegraphics[width=.3\textwidth]{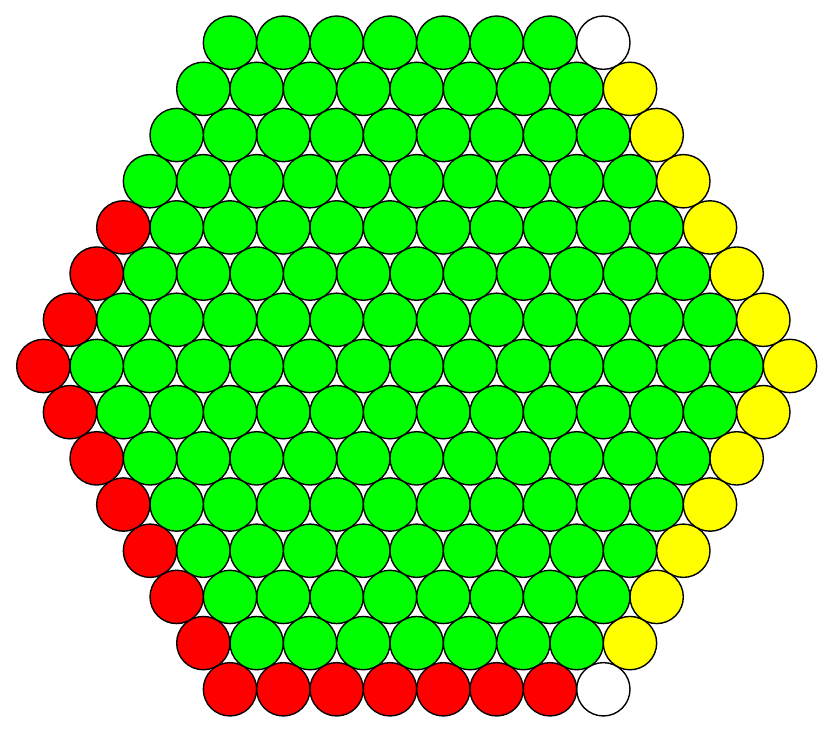}
        \includegraphics[width=.3\textwidth]{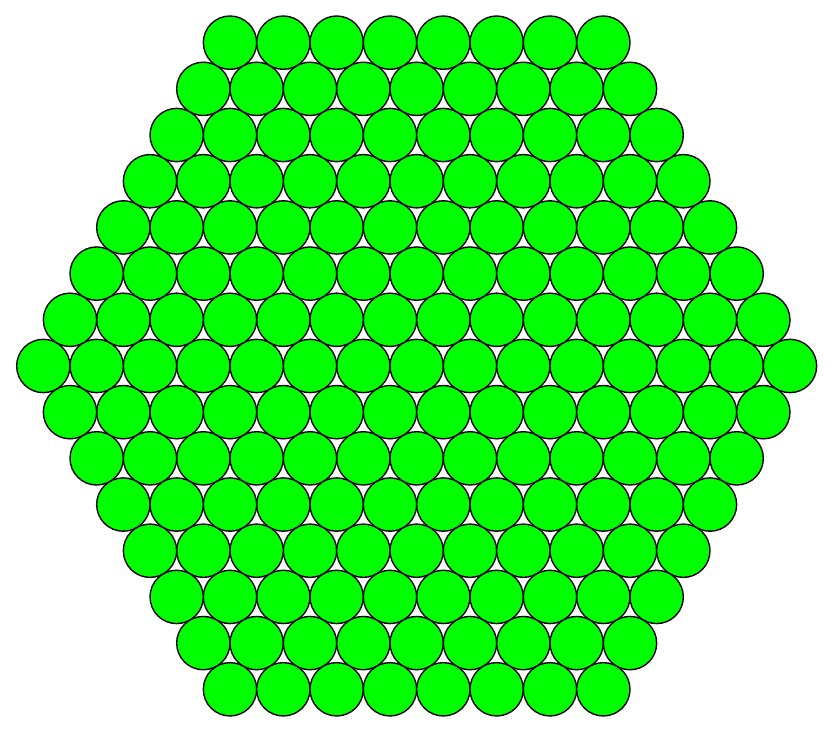}
      }
    \caption{Dynamics where all cells evolve according to the rules. The center cell determines the final outcome.}
    \label{Evolution}
  \end{center}
\end{figure*}

\section{ODEs with acyclic dependency graph}

As an aside note we will mention that our main result is also valid for dynamical systems with differential equations. We now consider ODEs that have an acyclic dependency graph, that is, no feedback loops in their structure.
For any ODE with acyclic dependency graph  we can see easily that its adjacency matrix $[f]$ is a strictly triangular matrix, i.e., with
zeros in the diagonal (otherwise, there would be loops in the graph).

\begin{thm}\label{ODE}
Consider an ODE with acyclic dependency graph, then there is a unique steady state and it is asymptotically stable.
\end{thm}
 
To proof this we need the following lemma.

\begin{lem}
Let $z:[0,\infty)\rightarrow \mathbb{R}$ be a continuous function such that $\lim_{t\rightarrow \infty}z(t)=z_0\in \mathbb{R}$. If $y'=z-ky,y(0)=y_0$ with $k>0$, then $\lim_{t\rightarrow \infty}y(t)=\frac{z_0}{k}$ (regardless of the value of $y_0$).
\end{lem}
\begin{proof}
First, notice that from $y'=z-ky$ we obtain $y(t)=y_0 e^{-kt}+\int_0^t{e^{k(s-t)}z(s)ds}$. Furthermore, $y(t)-\frac{z_0}{k}=y_0 e^{-kt}+\int_0^t{e^{k(s-t)}(z(s)-z_0)ds}+\int_0^t{e^{k(s-t)}z_0ds}-\frac{z_0}{k}$. It is easy to show that $\lim_{t\rightarrow \infty}\int_0^t{e^{k(s-t)}z_0ds}-\frac{z_0}{k}=\lim_{t\rightarrow \infty} y_0 e^{-kt}=0$; then, we only need to show that $\lim_{t\rightarrow \infty}\int_0^t{e^{k(s-t)}(z(s)-z_0)ds}=0$. First, since $\lim_{t\rightarrow \infty}z(t)=z_0$, it follows that $z(s)$ is bounded (let us say by $M$) and $\lim_{t\rightarrow \infty}sup_{s\in [\frac{t}{2},t]}\{|z(s)-z_0|\}=0$. On the other hand we have
\begin{flushleft}
$|\int_0^t{e^{k(s-t)}(z(s)-z_0)ds}|\leq$
\end{flushleft}

$\int_0^{\frac{t}{2}}e^{k(s-t)}|z(s)-z_0|ds+\int_{\frac{t}{2}}^{t}e^{k(s-t)}|z(s)-z_0|ds$

$ \leq 2M\int_0^{\frac{t}{2}}e^{k(s-t)}ds+$

$sup_{s\in [\frac{t}{2},t]}\{|z(s)-z_0|\}\int_{\frac{t}{2}}^{t}e^{k(s-t)}ds$

$ = 2M\frac{e^{-kt/2}-e^{-kt}}{k}+$

$\ \ \ sup_{s\in [\frac{t}{2},t]}\{|z(s)-z_0|\}\frac{1-e^{-kt/2}}{k}$

$\leq 2M\frac{e^{-kt/2}}{k}+sup_{s\in [\frac{t}{2},t]}\{|z(s)-z_0|\}\frac{1}{k} $

The last expression converges to 0 as $t\rightarrow \infty$ and this finishes the proof.

\end{proof}

\begin{cor}
Consider $x:[0,\infty)\rightarrow \mathbb{R}^n$, $f:\mathbb{R}^n\rightarrow \mathbb{R}$ continuous such that $\lim_{t\rightarrow \infty}x(t)=x_0\in \mathbb{R}^n$. If $y'=f(x)-ky$ with $k>0$, then $\lim_{t\rightarrow \infty}y(t)=\frac{f(x_0)}{k}$.
\end{cor}

\begin{proof}
It is enough to consider $z(t)=f(x(t))$ and apply the lemma above.
\end{proof}

When modeling biological systems using ODEs, it is common for the functions to have natural decay; that is, they are of the form $x_i'=f(x)-k_ix_i$ for $i=1,\ldots,n$. The dependency graph of such an ODE is the graph with nodes $\{1,\ldots,n\}$ and an edge from $i$ to $j$ if $f_j$ depends on $x_i$.

\begin{proof}[Proof of theorem~\ref{ODE}]
Consider an ODE with natural decay. Without loss of generality, we consider that the adjacency matrix of the dependency graph is of the form
\[
\left[ \begin{array}{ccccc} 0 & 0  & 0 & \ldots & 0  \\ * & 0 & 0 & \ldots & 0 \\ * & * & 0 & \ldots & 0 \\ \vdots &    &   & \ddots  &   \\ * & * & \ldots & * & 0 \end{array} \right]
\]
That is, $f_i=f_i(x_1,\ldots,x_{i-1})$ ($f_i$ could depend on less variables). For $i=1$ we have $x_1'=f_1-k_1x_1$ where $f_1$ is constant. Then, $\lim_{t\rightarrow \infty}x_1(t)=s_1$, where $s_1=\frac{f_1}{k_1}$. For $i=2$ we have $x_2'=f_2(x_1)-k_2x_2$; then, by the corollary, we have that $\lim_{t\rightarrow \infty}x_2(t)=s_2$, where $s_2=\frac{f_2(s_1)}{k_2}$. By induction, if $\lim_{t\rightarrow \infty}(x_1(t),\ldots,x_{i-1}(t))=(s_1,\ldots,s_{i-1})$ and $x_i'=f_i(x_1,\ldots,x_{i-1})-k_ix_i$, then $\lim_{t\rightarrow \infty}x_i(t)=s_i$, where $s_i=\frac{f_i(s_1,\ldots,s_{i-1})}{k_i}$. At the end, we have $s=(s_1,\ldots,s_n)\in\mathbb{R}^n$ such that $\lim_{t\rightarrow \infty}x(t)=s$ (regardless of the value of $x(0)$). It follows that $s$ is the unique steady state and that it is asymptotically stable.
\end{proof}
 
\section{Discussion}
In this paper we have shown that it is fruitful to study the relationship between the structure and the dynamics
of discrete dynamical systems by looking at the algebraic properties of the mapping 
$\mathcal{P}\longrightarrow\mathcal{M}$ from the algebra of PDS to the algebra of adjacency matrices. 
With a very straightforward proof we have shown that dynamics is very simple in the absence of feedback loops. 
It is worth observing that this result, Corollary \ref{main}, 
could also have been obtained as a consequence of a more general
result that says that for the existence of more than one fixed point, a positive feedback loop is required
in the dependency graph, and for periodic orbits to exist a negative feedback is necessary (but not sufficient).
This result implies that in order to obtain more than one fixed point or periodic orbits, it is necessary that the
dependency graph of the system have feedback loops \cite{Soule}. However, the proof we have given here
of this same result is very simple and stems from a basic property of the mapping 
$\mathcal{P}\longrightarrow\mathcal{M}$ rather than complicated phase space arguments. It emphasizes our
belief that the proper framework for studying the relationship between structure and dynamics
of PDS is the algebra inherent in this mapping. 

Furthermore, we have shown that one can use Corollary \ref{main} to draw non-obvious conclusions about
a system of interest. This conclusion would be very difficult to arrive at through simulations, due to the
combinatorial complexity of the dynamics on a large grid. 


\end{document}